\documentclass[twoside,leqno,letterpaper,11pt]{article}
\usepackage{amssymb,amsmath,amscd,latexsym,epic,eepic,gastex}
\usepackage{url,paralist}
\usepackage{luca,macros}
\usepackage{ltexpprt}
\usepackage{times}

\newtheorem{remark}{Remark}

\newcommand{\myomit}[1]{{}}

\begin{document}

\title{The Complexity of Partial-observation Stochastic Parity Games With Finite-memory Strategies\thanks{
This research was supported by Austrian Science Fund (FWF) Grant No P23499- N23, FWF NFN Grant No S11407-N23 (RiSE), ERC Start grant (279307: Graph Games),
Microsoft Faculty Fellowship Award, 
NSF grants CNS 1049862 and CCF-1139011,
by NSF Expeditions in Computing project "ExCAPE: Expeditions in Computer 
Augmented Program Engineering", by BSF grant 9800096, and by gift from 
Intel.}
}
\author{
Krishnendu Chatterjee$^\dag$ \quad  Laurent Doyen$^{\S}$ \quad  Sumit Nain$^{\ddag}$ \quad Moshe Y. Vardi$^{\ddag}$ \\ 
\normalsize
  $\strut^\dag$ IST Austria \\
\normalsize  $\strut^\S$ CNRS, LSV, ENS Cachan \\
\normalsize
  $\strut^\ddag$ Rice University, USA 
}

\date{}
\maketitle

\begin{abstract}
We consider two-player partial-observation stochastic games on finite-state 
graphs where player~1 has partial observation and player~2 has perfect 
observation. 
The winning condition we study are $\omega$-regular conditions specified as 
parity objectives.
The qualitative-analysis problem given a partial-observation stochastic game 
and a parity objective asks whether there is a strategy to ensure that the 
objective is satisfied with probability~1 (resp. positive probability).
These qualitative-analysis problems are known to be undecidable.
However in many applications the relevant question is the existence of finite-memory
strategies, and the qualitative-analysis problems under finite-memory 
strategies was recently shown to be decidable in 2EXPTIME.
We improve the complexity and show that the qualitative-analysis problems for 
partial-observation stochastic parity games under finite-memory strategies are 
EXPTIME-complete; and also establish optimal (exponential) memory bounds for 
finite-memory strategies required for qualitative analysis. 
\end{abstract}

\section{Introduction}

\noindent{\em Games on graphs.}
Two-player stochastic games on finite graphs played for infinite rounds 
is central in many areas of computer science as they provide a natural setting 
to model nondeterminism and reactivity in the presence of randomness. 
In particular, infinite-duration games with omega-regular objectives are a 
fundamental tool in the analysis of many aspects of reactive systems  such as 
modeling, verification, refinement, and synthesis~\cite{AHK02,McN93}. 
For example, the standard approach to the synthesis problem for reactive 
systems reduces the problem to finding the winning strategy of a suitable 
game~\cite{PR89a}.
The most common approach to games assumes a setting with perfect information, 
where both players have complete knowledge of the state of the game.
In many settings, however, the assumption of perfect information is not valid and 
it is natural to allow an information asymmetry between the players,
such as, controllers with noisy sensors and software modules that expose partial 
interfaces~\cite{Reif84}.

\smallskip\noindent{\em Partial-observation stochastic games.}
Partial-observation stochastic games are played between two players (player~1 and player~2) 
on a graph with finite state space.
The game is played for infinitely many rounds where in each round either 
player~1 chooses a move or player~2 chooses a move, and the successor state is 
determined by a probabilistic transition function. 
Player~1 has partial observation where the state space is partitioned according to observations 
that she can observe i.e., given the current state, the player can only view the observation 
of the state (the partition the state belongs to), but not the precise state.
Player~2, the adversary to player~1, has perfect observation and can observe the precise state.

\smallskip\noindent{\em The class of $\omega$-regular objectives.}
An objective specifies the desired set of behaviors (or paths) for player~1. 
In verification and control of stochastic systems an objective is 
typically an $\omega$-regular set of paths. 
The class of $\omega$-regular languages extends classical regular languages to 
infinite strings, and provides a robust specification language to express
all commonly used specifications~\cite{Thomas97}.
In a parity objective, every state of the game is mapped to a non-negative 
integer priority and the goal is to ensure that the minimum priority visited 
infinitely often is even.
Parity objectives are a canonical way to define such $\omega$-regular 
specifications.
Thus partial-observation stochastic games with parity objective provide a 
general framework for analysis of stochastic reactive systems.

\smallskip\noindent{\em Qualitative and quantitative analysis.} 
Given a partial-observation stochastic game with a parity objective and a 
start state, the \emph{qualitative-analysis} problem asks whether the 
objective can be ensured with probability~1 (\emph{almost-sure winning}) or 
positive probability (\emph{positive winning}); whereas the more general 
\emph{quantitative-analysis} problem asks whether the objective can be 
satisfied with probability at least $\lambda$ for a given threshold 
$\lambda \in (0,1)$.

\smallskip\noindent{\em Previous results.}
The quantitative analysis problem for partial-observation stochastic games with parity objectives
is undecidable, even for the very special case of probabilistic automata with 
reachability objectives~\cite{Paz71}.
The qualitative-analysis problems for  partial-observation stochastic games with parity objectives 
are also undecidable~\cite{BBG08}, even for probabilistic automata.
In many practical applications, however, the more relevant question is the 
existence of finite-memory strategies.
The quantitative analysis problem remains undecidable for finite-memory 
strategies, even for probabilistic automata~\cite{Paz71}.
The qualitative-analysis problems for partial-observation stochastic parity 
games were shown to be decidable with 2EXPTIME complexity for finite-memory 
strategies~\cite{NV13};
and the exact complexity of the problems was open which we settle in this work.

\smallskip\noindent{\em Our contributions.}
Our contributions are as follows: for the qualitative-analysis problems for 
partial-observation stochastic parity games under finite-memory strategies we 
show that
(i)~the problems are EXPTIME-complete; and 
(ii)~if there is a finite-memory almost-sure (resp. positive) winning strategy,
then there is a strategy that uses at most exponential memory (matching the 
exponential lower bound known for the simpler case of reachability and safety 
objectives).
Thus we establish both optimal computational and strategy complexity results.
Moreover, once a finite-memory strategy is fixed for player~1, we obtain a finite-state 
perfect-information Markov decision process (MDP) for player~2 where finite-memory 
is as powerful as infinite-memory~\cite{CY95}.
Thus our results apply to both cases where player~2 has infinite-memory or restricted
to finite-memory  strategies.

\smallskip\noindent{\em Technical contribution.}
The 2EXPTIME upper bound of~\cite{NV13} is achieved via a reduction to the emptiness 
problem of alternating parity tree automata.
The reduction of~\cite{NV13} to alternating tree automata is exponential as it requires 
enumeration of the end components and recurrent classes that can arise after fixing strategies.
We present a polynomial reduction, which is achieved in two steps.
The first step is as follows: a~\emph{local gadget-based} reduction (that transforms every 
probabilistic state to a local gadget of deterministic states) for perfect-observation 
stochastic games to perfect-observation deterministic games for parity objectives was
presented in~\cite{CJH03,Cha-Thesis}.
This gadget, however, requires perfect observation for both players.
We extend this reduction and present a local gadget-based polynomial reduction of 
partial-observation stochastic games to three-player partial-observation deterministic 
games, where player~1 has partial observation, the other two players have perfect observation, 
and player~3 is helpful to player~1.
The crux of the proof is to show that the local reduction allows to infer properties about
recurrent classes and end components (which are global properties).
In the second step we present a polynomial reduction of the three-player games problem 
to the emptiness problem of alternating tree automata. 
We also remark that the new model of three-player games we introduce for the intermediate 
step of the reduction maybe also of independent interest for modeling of other applications.

\smallskip\noindent{\em Related works.}
The undecidability of the qualitative-analysis problem for partial-observation stochastic 
parity games with infinite-memory strategies follows from~\cite{BBG08}.
For partially observable Markov decision processes (POMDPs), which is a special case of 
partial-observation stochastic games where player~2 does not have any choices, the qualitative-analysis
problem for parity objectives with finite-memory strategies was shown to be EXPTIME-complete~\cite{CCT13}.
For partial-observation stochastic games the almost-sure winning problem was shown to be 
EXPTIME-complete for B\"uchi objectives (both for finite-memory and infinite-memory strategies)~\cite{CDHR07,CD12}.
Finally, for partial-observation stochastic parity games the almost-sure winning problem
under finite-memory strategies was shown to be decidable in 2EXPTIME in~\cite{NV13}.

\smallskip\noindent{\em Summary and discussion.} 
The results for the qualitative analysis of various models of partial-observation 
stochastic parity games with finite-memory strategies for player~1 is summarized
in Table~\ref{tab:complexity}.
We explain the results of the table.
The results of the first row follows from~\cite{CCT13} and the results for the second
row are the results of our contributions.
In the most general case both players have partial observation~\cite{BGG09}.
If we consider partial-observation stochastic games where both players have partial observation,
then the results of the table are derived as follows:
(a)~If we consider infinite-memory strategies for player~2, then the problem remains undecidable
as when player~1 is non-existent we obtain POMDPs as a special case.
The non-elementary lower bound follows from the results of~\cite{CD12} where the lower
bound was shown for reachability objectives where finite-memory strategies suffice for 
player~1 (against both finite and infinite-memory strategies for player~2).
(b)~If we consider finite-memory strategies for player~2, then the decidability of
the problem is open, but we obtain the non-elementary lower bound on memory from the
results of~\cite{CD12} for reachability objectives.

\begin{table}[h]
\begin{center}
\begin{tabular}{|c|c|c|}
\hline
Game Models & Complexity & Memory bounds\\
\hline
\hline
POMDPs & \ EXPTIME-complete~\cite{CCT13} \ & Exponential~\cite{CCT13} \\
\hline
Player~1 partial and player~2 perfect & {\bf EXPTIME-complete} & {\bf Exponential} \\
\ (finite- or infinite-memory for player~2) \ & & \\
\hline
Both players partial  & Undecidable~\cite{BBG08} & \ Non-elementary~\cite{CD12} \  \\
infinite-memory for player~2 &  & (Lower bound) \\
\hline
Both players partial & Open (??) & Non-elementary~\cite{CD12} \\
finite-memory for player~2 &  & (Lower bound) \\
\hline
\hline
\end{tabular}
\end{center}
\caption{Complexity and memory bounds for qualitative analysis of partial-observation stochastic 
parity games with finite-memory strategies for player~1. The new results are boldfaced.}\label{tab:complexity}
\end{table}


\section{Partial-observation Stochastic Parity Games}\label{sec:partial_stoch}
We consider partial-observation stochastic parity games
where player~1 has partial observation and player~2 has perfect 
observation.
We consider parity objectives, and for almost-sure winning
under finite-memory strategies for player~1 present a polynomial reduction 
to sure winning in three-player parity games where player~1 has 
partial observation, player~3 has perfect observation and is helpful 
towards player~1, and player~2 has perfect observation and is 
adversarial to player~1.
A similar reduction also works for positive winning.
We then show how to solve the  sure-winning problem for 
three-player games using alternating parity tree automata.
Thus the steps are as follows:
\begin{enumerate}

\item Reduction of partial-observation stochastic parity games 
for almost-sure winning with finite-memory strategies to 
three-player parity games sure-winning problem (with player~1 partial, other 
two perfect, player~1 and player~3 existential, and player~2 adversarial).

\item Solving the sure winning problem for three-player parity games 
using alternating parity tree automata.

\end{enumerate}
In this section we present the details of the first step.
The second step is given in the following section.

\subsection{Basic definitions} 
We start with basic definitions related to partial-observation stochastic 
parity games.

\smallskip\noindent{\em Partial-observation stochastic games.}
We consider slightly different notation (though equivalent) to 
the classical definitions, but the slightly different notation helps for
more elegant and explicit reduction.
We consider partial-observation stochastic games as a tuple 
$G=(\VA,\VB,\VR, A_1, \trans, E, \Obs,\obs)$ as follows:
$S=\VA \cup \VB \cup \VR$ is the state space partitioned into player-1 states ($\VA$),
player-2 states ($\VB$), and probabilistic states ($\VR$);
and $A_1$ is a finite set of actions for player~1.
Since player~2 has perfect observation, she chooses edges instead
of actions.
The transition function is as follows: 
$\trans: \VA \times A_1\to S_2$ that given a player-1 state in $\VA$ 
and an action in $A_1$ gives the next state in $\VB$ (which belongs to 
player~2); and $\trans:\VR \to \dist(\VA)$ given a probabilistic state
gives the probability distribution over the set of player-1 
states.
The set of edges is as follows:
$E= \set{(s,t) \mid s \in \VR, t \in \VA, \trans(s)(t) > 0} \cup E'$,
where $E' \subseteq \VB \times \VR$.
The observation set $\Obs$ and observation mapping $\obs$ are standard,
i.e., $\obs:S \to \Obs$.
Note that player~1 plays after every three steps (every move of 
player~1 is followed by a move of player~2, then a probabilistic
choice).
In other words, first player~1 chooses an action, then player~2 chooses
an edge, and then there is a probability distribution over states 
where player~1 again chooses and so on.

\smallskip\noindent{\em Three-player non-stochastic turn-based games.}
We consider three-player partial-observation (non-stochastic turn-based) games 
as a tuple $G=(\VA,\VB,\VC, A_1, \trans, E, \Obs,\obs)$ as follows:
$S$ is the state space partitioned into player-1 states ($\VA$),
player-2 states ($\VB$), and player-3 states ($\VC$);
and $A_1$ is a finite set of actions for player~1.
The transition function is as follows: 
$\trans: \VA \times A_1\to S_2$ that given a player-1 state in $\VA$ 
and an action in $A_1$ gives the next state (which belongs to player~2).
The set of edges is as follows:
$E \subseteq (\VB \cup \VC) \times S$.
Hence in these games player~1 chooses an action, and the other 
players have perfect observation and choose edges.
We only consider the sub-class where player~1 plays in 
every $k$-steps, for a fixed $k$.
The observation set $\Obs$ and observation mapping $\obs$ are again 
standard.

\smallskip\noindent{\em Plays and strategies.}
A \emph{play} in a partial-observation stochastic game is an infinite sequence of states 
$s_0 s_1 s_2 \ldots$ such that the following conditions hold for all $i \geq 0$: 
(i)~if $s_i \in \VA$, then there exists $a_i \in A_1$ such that $s_{i+1}=\trans(s_i,a_i)$;
and (ii)~if $s_i \in (\VB \cup \VR)$, then $(s_i,s_{i+1}) \in E$.
The function $\obs$ is extended to sequences $\rho = s_0 \dots s_n$ of states 
in the natural way, namely $\obs(\rho) = \obs(s_0) \dots \obs(s_n)$.
A strategy for a player is a recipe to extend the prefix of a play.
Formally, player-1 strategies are functions $\straa: S^* \cdot \VA \to A_1$; and 
player-2 (and analogously player-3 strategies) are functions: 
$\strab: S^* \cdot \VB \to S$
such that for all $w \in S^*$ and $s \in \VB$ we have $(s,\strab(w \cdot s)) \in E$.
We consider only observation-based strategies for player~1, i.e., 
for two play prefixes $\rho$ and $\rho'$ if the corresponding observation sequences match
($\obs(\rho)=\obs(\rho')$), then the strategy must choose the same action 
($\straa(\rho)=\straa(\rho')$); and the other players have all strategies.
The notations for three-player games are similar.

\smallskip\noindent{\em Finite-memory strategies.}
A player-1 strategy uses \emph{finite-memory} if it can be encoded
by a deterministic transducer $\tuple{\mem, m_0, \straa_u, \straa_n}$
where $\mem$ is a finite set (the memory of the strategy), 
$m_0 \in \mem$ is the initial memory value,
$\straa_u: \mem \times \Obs  \to \mem$ is the memory-update function, and 
$\straa_n: \mem \to A_1$ is the next-move function. 
The \emph{size} of the strategy is the number $\abs{\mem}$ of memory values.
If the current observation is $o$, and the current memory value is $m$,
then the strategy chooses the next action 
$\straa_n(m)$,
and the memory is updated to $\straa_u(m,o)$. 
Formally, $\tuple{\mem, m_0, \straa_u, \straa_n}$
defines the strategy $\straa$ such that 
$\straa(\rho\cdot s)=\straa_n(\widehat{\straa}_u(m_0, \obs(\rho)\cdot \obs(s))$
for all $\rho \in S^*$ and $s \in \VA$, where $\widehat{\straa}_u$ extends 
$\straa_u$ to sequences of observations as expected. 
This definition extends to infinite-memory strategies by dropping the 
assumption that the set $\mem$ is finite.

\smallskip\noindent{\em Parity objectives.}
An \emph{objective} for Player~$1$ in $G$ is a set $\varphi \subseteq S^\omega$ of 
infinite sequences of states. 
A play $\rho$ \emph{satisfies} the objective $\varphi$ if $\rho \in \varphi$.
For a play $\rho=s_0 s_1 \ldots$ we denote by $\Inf(\rho)$ the set of states 
that occur infinitely often in $\rho$, that is, 
$\Inf(\rho)=\{ s \mid s_j=s \text{ for infinitely many } j \text{'s}  \}$.
For $d \in \Nats$, let $p:S \to \{0,1,\ldots,d\}$ be a 
\emph{priority function}, 
which maps each state to a nonnegative integer priority.
The \emph{parity} objective $\Parity(p)$ requires that the minimum priority 
that occurs infinitely often be even.
Formally, $\Parity(p)=\set{\rho \mid \min\set{ p(s) \mid s \in \Inf(\rho)} 
\mbox{ is even} }$.
Parity objectives are a canonical way to express $\omega$-regular objectives~\cite{Thomas97}.

\smallskip\noindent{\em Almost-sure winning and positive winning.}
An \emph{event} is a measurable set of plays.
For a partial-observation stochastic game, given strategies $\straa$ and $\strab$ 
for the two players, the probabilities of events are uniquely defined~\cite{Var85}. 
For a parity objective~$\Parity(p)$, we denote by $\Prb_{s}^{\straa,\strab}(\Parity(p))$ 
the probability that $\Parity(p)$ is satisfied by the play obtained from the starting state $s$ 
when the strategies $\straa$ and $\strab$ are used.
The \emph{almost-sure} (resp. \emph{positive}) winning problem under finite-memory strategies 
asks, given a partial-observation stochastic game, a parity objective $\Parity(p)$, and a 
starting state $s$,  whether there exists a finite-memory observation-based 
strategy $\straa$ for player~1 such that against all strategies $\strab$ for player~2 
we have  $\Prb_{s}^{\straa,\strab}(\Parity(p))=1$ (resp. $\Prb_{s}^{\straa,\strab}(\Parity(p))>0$).
The almost-sure and positive winning problems are also referred to as the qualitative-analysis problems
for stochastic games.

\smallskip\noindent{\em Sure winning in three-player games.}
In three-player games once the starting state $s$ and strategies $\straa, \strab$, and $\strac$
of the three players are fixed we obtain a unique play, which we denote as $\rho_s^{\straa,\strab,\strac}$. 
In three-player games we consider the following \emph{sure} winning problem:
given a parity objective $\Parity(p)$, sure winning is ensured if 
there exists a finite-memory observation-based strategy $\straa$ for 
player~1, such that in the two-player perfect-observation game 
obtained after fixing $\straa$, player~3 can ensure the parity 
objective against all strategies of player~2. 
Formally, the sure winning problem asks whether there exist a finite-memory observation-based 
strategy $\straa$ for player~1 and a strategy $\strac$ for player~3, such that for all 
strategies $\strab$ for player~2 we have  $\rho_s^{\straa,\strab,\strac} \in \Parity(p)$.

\begin{remark}[Equivalence with standard model]
We remark that for the model of partial-observation stochastic games studied in literature
the two players simultaneously choose actions, and a probabilistic transition function 
determine the probability distribution of the next state. 
In our model, the game is turn-based and the probability distribution is chosen only in 
probabilistic states.
However, it follows from the results of~\cite{CDGH10} that the models are equivalent:
by the results of~\cite[Section~3.1]{CDGH10} the interaction of the players and probability 
can be separated without loss of generality; and~\cite[Theorem~4]{CDGH10} shows that in presence of 
partial observation, concurrent games can be reduced to turn-based games in polynomial time.
Thus the turn-based model where the moves of the players and stochastic interaction 
are separated is equivalent to the standard model.
Moreover, for a perfect-information player choosing an action is equivalent to choosing 
an edge in a turn-based game. 
Thus the model we consider is equivalent to the standard partial-observation game models.
\end{remark}

\begin{remark}[Pure and randomized strategies]\label{remark:strategies}
In this work we only consider pure strategies. 
In partial-observation games, randomized strategies are also relevant as they are more 
powerful than pure strategies. 
However, for finite-memory strategies the almost-sure and positive winning problem 
for randomized strategies can be reduced in polynomial time to the problem 
for finite-memory pure strategies~\cite{CD12,NV13}.
Hence without loss of generality we only consider pure strategies.
\end{remark}

\subsection{Reduction of partial-observation stochastic games to three-player games}
In this section we present a polynomial-time reduction for the almost-sure winning 
problem in partial-observation stochastic parity games to the sure winning problem 
in three-player parity games.

\smallskip\noindent{\em Reduction.} 
Let us denote by $[d]$ the set $\set{0,1,\ldots,d}$.
Given a partial-observation stochastic parity game graph 
$G=(\VA,\VB,\VR, A_1, \trans, E, \Obs,\obs)$ 
with a parity objective defined by priority function 
$p: S \to [d]$ 
we construct a three-player game graph
$\ov{G}=(\ovVA,\ovVB,\ovVC, A_1, \ov{\trans}, \ov{E}, \Obs,\ov{\obs})$
together with priority function $\ov{p}$.
The construction is specified as follows.
\begin{enumerate}
\item For every nonprobabilistic state $s \in \VA \cup \VB$, there is a 
corresponding state $\overline{s} \in \overline{S}$ 
such that 
\begin{itemize}
\item~$\ov{s}\in\ovVA$ if $s\in\VA$, else $\ov{s} \in \ovVB$;
\item~$\ov{p}(\ov{s})=p(s)$ and $\ov{\obs}(\ov{s})=\obs(s)$; 
\item~$\ov{\trans}(\ov{s},a)=\ov{t}$ where $t=\trans(s,a)$, 
for $s \in \VA$ and $a\in A_1$; and 
\item $(\ov{s}, \ov{t}) \in \ov{E}$ iff $(s, t) \in E$, for $s \in \VB$.
\end{itemize}

\item 
Every probabilistic state  $s \in \VR$ is replaced by the gadget
shown in Figure~\ref{fig:gadget-even} and Figure~\ref{fig:gadget-odd}.
In the figure, square-shaped states are player-2 states 
(in~$\ovVB$), and circle-shaped (or ellipsoid-shaped) states are player-3 states
(in~$\ovVC$).
Formally, from the state~$\ov{s}$ with priority $p(s)$ and observation 
$\obs(s)$ (i.e., $\ov{p}(\ov{s})=p(s)$ and  $\ov{\obs}(\ov{s})=\obs(s)$) 
the players play the following three-step game in~$\overline{G}$.
\begin{itemize}
\item First, in state $\overline{s}$ player~2 chooses a successor  
$(\wt{s}, 2k)$, for $2k \in \set{0,1,\ldots,p(s)+1}$.

\item For every state $(\wt{s},2k)$, we have $\ov{p}((\wt{s},2k))=p(s)$ and 
$\ov{\obs}((\wt{s},2k))=\obs(s)$.
For $k\geq 1$, in state $(\wt{s}, 2k)$ player~3 chooses between
two successors:  
state $(\wh{s}, 2k-1)$ with priority $2k-1$ and same observation as $s$, or
state $(\wh{s}, 2k)$ with priority $2k$ and same observation as $s$, 
(i.e., $\ov{p}((\wh{s}, 2k-1))=2k-1$, $\ov{p}((\wh{s}, 2k))=2k$, and
$\ov{\obs}((\wh{s}, 2k-1)) =\ov{\obs}((\wh{s}, 2k)) =\obs(s)$).
The state $(\wt{s},0)$ has only one successor $(\wh{s},0)$, with
$\ov{p}((\wh{s},0))=0$ and $\ov{\obs}((\wh{s}, 0)) =\obs(s)$.

\item 
Finally, in each state $(\wh{s}, k)$ the choice is between all
states $\overline{t}$ such that $(s, t) \in E$, and it belongs to
player~3 (i.e., in $\ovVC$) if $k$ is odd, and to player~2 (i.e., in 
$\ovVB$) if $k$ is even.
Note that every state in the gadget has the same observation as the original state.
\end{itemize}
\end{enumerate}
We denote by $\ov{G}=\tras(G)$ the three-player game, where player~1 has 
partial-observation, and both player~2 and player~3 have perfect-observation, 
obtained from a partial-observation stochastic game. Also observe that in 
$\ov{G}$ there are exactly four steps between two player~1 moves.

\smallskip\noindent{\em Observation sequence mapping.}
Note that since in our partial-observation games first player~1 plays,
then player~2, followed by probabilistic states, repeated ad infinitum, wlog, 
we can assume that for every observation $o\in\Obs$ we have either 
(i)~$\obs^{-1}(o) \subseteq \VA$; or
(ii)~$\obs^{-1}(o) \subseteq \VB$; or
(i)~$\obs^{-1}(o) \subseteq \VR$.
Thus we partition the observations as $\Obs_1$, $\Obs_2$, and $\Obs_P$.
Given an observation sequence $\kappa=o_0 o_1 o_2 \ldots o_n$ in $G$ corresponding
to a finite prefix of a play, we inductively define the sequence $\ov{\kappa}=\ov{h}(\kappa)$ 
in $\ov{G}$ as follows: 
(i)~$\ov{h}(o_0)=o_0$ if $o_0 \in \Obs_1 \cup \Obs_2$, else $o_0 o_0 o_0$; 
(ii)~$\ov{h}(o_0 o_1 \ldots o_n)= \ov{h}(o_0 o_1 \ldots o_{n-1}) o_n$  
if $o_n \in \Obs_1 \cup \Obs_2$, else $\ov{h}(o_0 o_1 \ldots o_{n-1}) o_n o_n o_n$.
Intuitively the mapping takes care of the two extra steps of the gadgets 
introduced for probabilistic states. 
The mapping is a bijection, and hence given an observation sequence $\ov{\kappa}$ of a 
play prefix in $\ov{G}$ we consider the inverse play prefix $\kappa=\ov{h}^{-1}(\ov{\kappa})$
such that $\ov{h}(\kappa)=\ov{\kappa}$.

\smallskip\noindent{\em Strategy mapping.} 
Given an observation-based strategy $\ov{\straa}$ in $\ov{G}$ we consider 
a strategy $\straa=\tras(\ov{\straa})$ as follows: 
for an observation sequence $\kappa$ corresponding to a play prefix in 
$G$ we have $\straa(\kappa)=\ov{\straa}(\ov{h}(\kappa))$.
The strategy $\straa$ is observation-based (since $\ov{\straa}$ is 
observation-based).
The inverse mapping $\tras^{-1}$ of strategies from $G$ to $\ov{G}$ 
is analogous. 
Note that for $\straa$ in $G$ we have $\tras(\tras^{-1}(\straa))=\straa$.
Let $\ov{\straa}$ be a finite-memory strategy  with memory $\mem$ for player~1 in 
the game $\ov{G}$. 
The strategy $\ov{\straa}$ can be considered as a memoryless 
strategy, denoted as $\ov{\straa}^*=\MemLess(\ov{\straa})$, in 
$\ov{G} \times \mem$ (the synchronous
product of $\ov{G}$ with $\mem$).
Given a strategy (pure memoryless) $\ov{\strab}$ for player~2 in the 
$2$-player game $\ov{G} \times \mem$, a strategy $\strab =\tras(\ov{\strab})$ 
in the partial-observation stochastic game $G\times \mem$ is defined as follows:
\[
\strab((s,m)) = (t,m'), \text{ if and only if } \ov{\strab}((\ov{s},m))=(\ov{t},m'); 
\text{ for all } s \in \VB.
\]  

\smallskip\noindent{\em End component and the key property.}
Given an MDP, a set $U$ is an end component in the MDP if the 
sub-graph induced by $U$ is strongly connected, and for all 
probabilistic states in $U$ all out-going edges end up in $U$ 
(i.e., $U$ is closed for probabilistic states). 
The key property about MDPs that is used in our proofs 
is a result established by~\cite{CY95,luca-thesis} that given
an MDP, for all strategies, with probability~1 the set of 
states visited infinitely often is an end component. 
The key property allows us to analyze end components of 
MDPs and from properties of the end component conclude 
properties about all strategies.

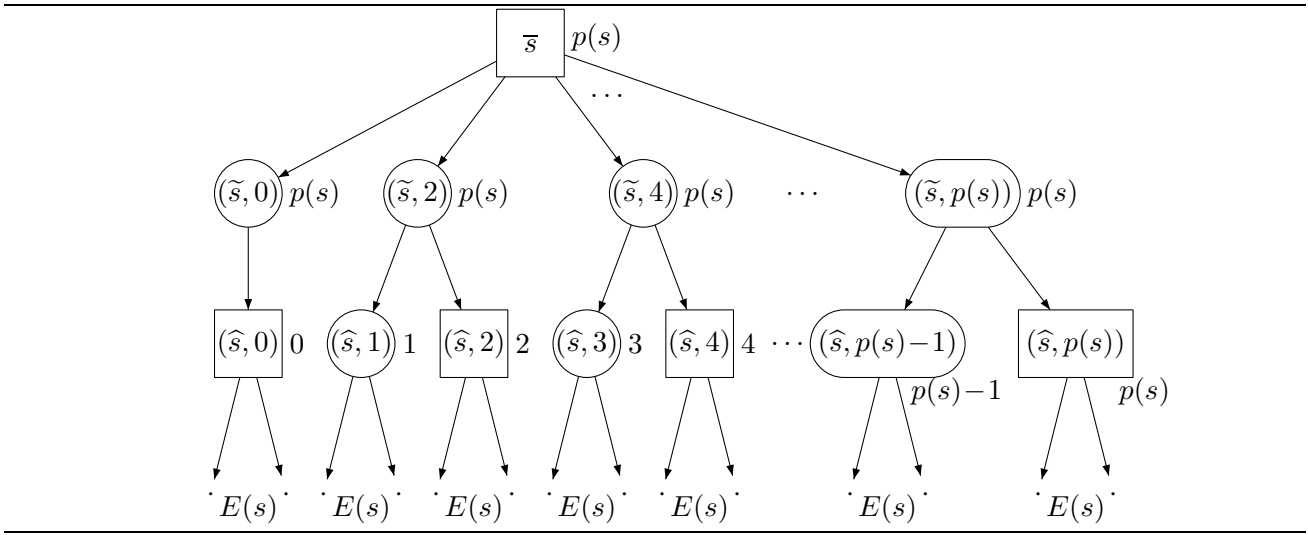
\begin{figure}[!ptb]
  \begin{center}
    \hrule
    \begin{picture}(122,70)(0,0)
\put(-3,0)
{

\gasset{Nw=9,Nh=9,Nmr=4.5,rdist=1, loopdiam=6}


\node[Nmarks=n, Nmr=0](s)(47.5,65){$\ov{s}$}
\nodelabel[ExtNL=y, NLangle=5, NLdist=1](s){$p(s)$}

\node[Nframe=n](dots)(58,58){$\dots$}

\node[Nmarks=n](t0)(10,45){$(\wt{s},0)$}
\nodelabel[ExtNL=y, NLangle=0, NLdist=1](t0){$p(s)$}
\node[Nmarks=n](t2)(32.5,45){$(\wt{s},2)$}
\nodelabel[ExtNL=y, NLangle=0, NLdist=1](t2){$p(s)$}
\node[Nmarks=n](t4)(62.5,45){$(\wt{s},4)$}
\nodelabel[ExtNL=y, NLangle=0, NLdist=1](t4){$p(s)$}

\node[Nframe=n](dots)(84,45){$\dots$}

\node[Nmarks=n, Nadjust=w, Nmr=8](tp)(105,45){$(\wt{s},p(s))$}
\nodelabel[ExtNL=y, NLangle=0, NLdist=1](tp){$p(s)$}

\node[Nmarks=n, Nmr=0](u0)(10,25){$(\wh{s},0)$}
\nodelabel[ExtNL=y, NLangle=0, NLdist=1](u0){$0$}
\node[Nmarks=n](u1)(25,25){$(\wh{s},1)$}
\nodelabel[ExtNL=y, NLangle=0, NLdist=1](u1){$1$}
\node[Nmarks=n, Nmr=0](u2)(40,25){$(\wh{s},2)$}
\nodelabel[ExtNL=y, NLangle=0, NLdist=1](u2){$2$}
\node[Nmarks=n](u3)(55,25){$(\wh{s},3)$}
\nodelabel[ExtNL=y, NLangle=0, NLdist=1](u3){$3$}
\node[Nmarks=n, Nmr=0](u4)(70,25){$(\wh{s},4)$}
\nodelabel[ExtNL=y, NLangle=0, NLdist=1](u4){$4$}

\node[Nframe=n](dots)(82,25){$\dots$}

\node[Nmarks=n, Nadjust=w](up)(95,25){$(\wh{s},p(s)\!-\!1)$}
\nodelabel[ExtNL=y, NLangle=325, NLdist=0](up){$p(s)\!-\!1$}

\node[Nmarks=n, Nmr=0, Nadjust=w](uq)(120,25){$(\wh{s},p(s))$}
\nodelabel[ExtNL=y, NLangle=325, NLdist=0](uq){$p(s)$}

\node[Nframe=n, Nw=4, Nh=4, Nmarks=n](dummy)(5,5){$\cdot$}
\drawedge[ELpos=50, ELside=r, curvedepth=0](u0,dummy){}

\node[Nframe=n, Nw=5, Nh=5, Nmarks=n](dummy)(15,5){$\cdot$}
\drawedge[ELpos=50, ELside=r, curvedepth=0](u0,dummy){}

\node[Nmarks=n, Nframe=n](label)(10,3){$E(s)$}

\node[Nframe=n, Nw=4, Nh=4, Nmarks=n](dummy)(20,5){$\cdot$}
\drawedge[ELpos=50, ELside=r, curvedepth=0](u1,dummy){}

\node[Nframe=n, Nw=5, Nh=5, Nmarks=n](dummy)(30,5){$\cdot$}
\drawedge[ELpos=50, ELside=r, curvedepth=0](u1,dummy){}

\node[Nmarks=n, Nframe=n](label)(25,3){$E(s)$}

\node[Nframe=n, Nw=4, Nh=4, Nmarks=n](dummy)(35,5){$\cdot$}
\drawedge[ELpos=50, ELside=r, curvedepth=0](u2,dummy){}

\node[Nframe=n, Nw=5, Nh=5, Nmarks=n](dummy)(45,5){$\cdot$}
\drawedge[ELpos=50, ELside=r, curvedepth=0](u2,dummy){}

\node[Nmarks=n, Nframe=n](label)(40,3){$E(s)$}

\node[Nframe=n, Nw=4, Nh=4, Nmarks=n](dummy)(50,5){$\cdot$}
\drawedge[ELpos=50, ELside=r, curvedepth=0](u3,dummy){}

\node[Nframe=n, Nw=5, Nh=5, Nmarks=n](dummy)(60,5){$\cdot$}
\drawedge[ELpos=50, ELside=r, curvedepth=0](u3,dummy){}

\node[Nmarks=n, Nframe=n](label)(55,3){$E(s)$}

\node[Nframe=n, Nw=4, Nh=4, Nmarks=n](dummy)(65,5){$\cdot$}
\drawedge[ELpos=50, ELside=r, curvedepth=0](u4,dummy){}

\node[Nframe=n, Nw=5, Nh=5, Nmarks=n](dummy)(75,5){$\cdot$}
\drawedge[ELpos=50, ELside=r, curvedepth=0](u4,dummy){}

\node[Nmarks=n, Nframe=n](label)(70,3){$E(s)$}

\node[Nframe=n, Nw=4, Nh=4, Nmarks=n](dummy)(90,5){$\cdot$}
\drawedge[ELpos=50, ELside=r, curvedepth=0](up,dummy){}

\node[Nframe=n, Nw=5, Nh=5, Nmarks=n](dummy)(100,5){$\cdot$}
\drawedge[ELpos=50, ELside=r, curvedepth=0](up,dummy){}

\node[Nmarks=n, Nframe=n](label)(95,3){$E(s)$}

\node[Nframe=n, Nw=4, Nh=4, Nmarks=n](dummy)(115,5){$\cdot$}
\drawedge[ELpos=50, ELside=r, curvedepth=0](uq,dummy){}

\node[Nframe=n, Nw=5, Nh=5, Nmarks=n](dummy)(125,5){$\cdot$}
\drawedge[ELpos=50, ELside=r, curvedepth=0](uq,dummy){}

\node[Nmarks=n, Nframe=n](label)(120,3){$E(s)$}

\drawedge[ELpos=50, ELside=r, curvedepth=0](s,t0){}
\drawedge[ELpos=50, ELside=l, curvedepth=0](s,t2){}

\drawedge[ELpos=50, ELside=r, curvedepth=0](s,t4){}
\drawedge[ELpos=50, ELside=l, curvedepth=0](s,tp){}

\drawedge[ELpos=50, ELside=r, curvedepth=0](t0,u0){}

\drawedge[ELpos=50, ELside=r, curvedepth=0](t2,u1){}
\drawedge[ELpos=50, ELside=r, curvedepth=0](t2,u2){}

\drawedge[ELpos=50, ELside=r, curvedepth=0](t4,u3){}
\drawedge[ELpos=50, ELside=r, curvedepth=0](t4,u4){}

\drawedge[ELpos=50, ELside=r, curvedepth=0](tp,up){}
\drawedge[ELpos=50, ELside=r, curvedepth=0](tp,uq){}
}

\end{picture}
    \hrule
    \caption{Reduction gadget when $p(s)$ is even. \label{fig:gadget-even}}
  \end{center}
\end{figure}

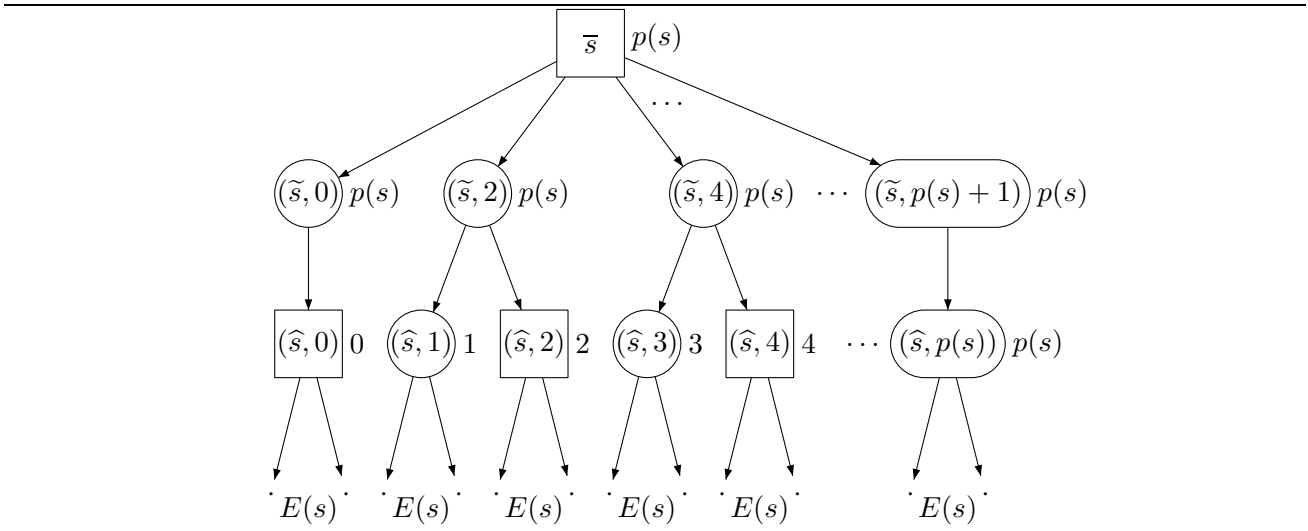
\begin{figure}[!ptb]
  \begin{center}
    \hrule
    \begin{picture}(112,70)(0,0)

\gasset{Nw=9,Nh=9,Nmr=4.5,rdist=1, loopdiam=6}


\node[Nmarks=n, Nmr=0](s)(47.5,65){$\ov{s}$}
\nodelabel[ExtNL=y, NLangle=5, NLdist=1](s){$p(s)$}

\node[Nframe=n](dots)(58,57){$\dots$}

\node[Nmarks=n](t0)(10,45){$(\wt{s},0)$}
\nodelabel[ExtNL=y, NLangle=0, NLdist=1](t0){$p(s)$}
\node[Nmarks=n](t2)(32.5,45){$(\wt{s},2)$}
\nodelabel[ExtNL=y, NLangle=0, NLdist=1](t2){$p(s)$}
\node[Nmarks=n](t4)(62.5,45){$(\wt{s},4)$}
\nodelabel[ExtNL=y, NLangle=0, NLdist=1](t4){$p(s)$}

\node[Nframe=n](dots)(80,45){$\dots$}

\node[Nmarks=n, Nadjust=w, Nmr=8](tp)(95,45){$(\wt{s},p(s)+1)$}
\nodelabel[ExtNL=y, NLangle=0, NLdist=1](tp){$p(s)$}

\node[Nmarks=n, Nmr=0](u0)(10,25){$(\wh{s},0)$}
\nodelabel[ExtNL=y, NLangle=0, NLdist=1](u0){$0$}
\node[Nmarks=n](u1)(25,25){$(\wh{s},1)$}
\nodelabel[ExtNL=y, NLangle=0, NLdist=1](u1){$1$}
\node[Nmarks=n, Nmr=0](u2)(40,25){$(\wh{s},2)$}
\nodelabel[ExtNL=y, NLangle=0, NLdist=1](u2){$2$}
\node[Nmarks=n](u3)(55,25){$(\wh{s},3)$}
\nodelabel[ExtNL=y, NLangle=0, NLdist=1](u3){$3$}
\node[Nmarks=n, Nmr=0](u4)(70,25){$(\wh{s},4)$}
\nodelabel[ExtNL=y, NLangle=0, NLdist=1](u4){$4$}

\node[Nframe=n](dots)(84,25){$\dots$}

\node[Nmarks=n, Nadjust=w](up)(95,25){$(\wh{s},p(s))$}
\nodelabel[ExtNL=y, NLangle=0, NLdist=1](up){$p(s)$}

\node[Nframe=n, Nw=4, Nh=4, Nmarks=n](dummy)(5,5){$\cdot$}
\drawedge[ELpos=50, ELside=r, curvedepth=0](u0,dummy){}

\node[Nframe=n, Nw=5, Nh=5, Nmarks=n](dummy)(15,5){$\cdot$}
\drawedge[ELpos=50, ELside=r, curvedepth=0](u0,dummy){}

\node[Nmarks=n, Nframe=n](label)(10,3){$E(s)$}

\node[Nframe=n, Nw=4, Nh=4, Nmarks=n](dummy)(20,5){$\cdot$}
\drawedge[ELpos=50, ELside=r, curvedepth=0](u1,dummy){}

\node[Nframe=n, Nw=5, Nh=5, Nmarks=n](dummy)(30,5){$\cdot$}
\drawedge[ELpos=50, ELside=r, curvedepth=0](u1,dummy){}

\node[Nmarks=n, Nframe=n](label)(25,3){$E(s)$}

\node[Nframe=n, Nw=4, Nh=4, Nmarks=n](dummy)(35,5){$\cdot$}
\drawedge[ELpos=50, ELside=r, curvedepth=0](u2,dummy){}

\node[Nframe=n, Nw=5, Nh=5, Nmarks=n](dummy)(45,5){$\cdot$}
\drawedge[ELpos=50, ELside=r, curvedepth=0](u2,dummy){}

\node[Nmarks=n, Nframe=n](label)(40,3){$E(s)$}

\node[Nframe=n, Nw=4, Nh=4, Nmarks=n](dummy)(50,5){$\cdot$}
\drawedge[ELpos=50, ELside=r, curvedepth=0](u3,dummy){}

\node[Nframe=n, Nw=5, Nh=5, Nmarks=n](dummy)(60,5){$\cdot$}
\drawedge[ELpos=50, ELside=r, curvedepth=0](u3,dummy){}

\node[Nmarks=n, Nframe=n](label)(55,3){$E(s)$}

\node[Nframe=n, Nw=4, Nh=4, Nmarks=n](dummy)(65,5){$\cdot$}
\drawedge[ELpos=50, ELside=r, curvedepth=0](u4,dummy){}

\node[Nframe=n, Nw=5, Nh=5, Nmarks=n](dummy)(75,5){$\cdot$}
\drawedge[ELpos=50, ELside=r, curvedepth=0](u4,dummy){}

\node[Nmarks=n, Nframe=n](label)(70,3){$E(s)$}

\node[Nframe=n, Nw=4, Nh=4, Nmarks=n](dummy)(90,5){$\cdot$}
\drawedge[ELpos=50, ELside=r, curvedepth=0](up,dummy){}

\node[Nframe=n, Nw=5, Nh=5, Nmarks=n](dummy)(100,5){$\cdot$}
\drawedge[ELpos=50, ELside=r, curvedepth=0](up,dummy){}

\node[Nmarks=n, Nframe=n](label)(95,3){$E(s)$}

\drawedge[ELpos=50, ELside=r, curvedepth=0](s,t0){}
\drawedge[ELpos=50, ELside=l, curvedepth=0](s,t2){}

\drawedge[ELpos=50, ELside=r, curvedepth=0](s,t4){}
\drawedge[ELpos=50, ELside=l, curvedepth=0](s,tp){}

\drawedge[ELpos=50, ELside=r, curvedepth=0](t0,u0){}

\drawedge[ELpos=50, ELside=r, curvedepth=0](t2,u1){}
\drawedge[ELpos=50, ELside=r, curvedepth=0](t2,u2){}

\drawedge[ELpos=50, ELside=r, curvedepth=0](t4,u3){}
\drawedge[ELpos=50, ELside=r, curvedepth=0](t4,u4){}

\drawedge[ELpos=50, ELside=r, curvedepth=0](tp,up){}


\end{picture}
    \hrule
    \caption{Reduction gadget when $p(s)$ is odd. \label{fig:gadget-odd}}
  \end{center}
\end{figure}

\smallskip\noindent{\em The key lemma.} We are now ready to present our main
lemma that establishes the correctness of the reduction. 
Since the proof of the lemma is long we split the proof into two parts.

\begin{lemma}\label{lemm:par-reduction1}
Given a partial-observation stochastic parity game $G$ with parity objective $\Parity(p)$,
let $\ov{G}=\tras(G)$ be the three-player game with the modified parity objective $\Parity(\ov{p})$
obtained by our reduction. 
Consider a finite-memory strategy $\ov{\straa}$ with memory $\mem$ for player~1
in $\ov{G}$.
Let us denote by $\ov{G}_{\ov{\straa}}$ the perfect-observation two-player game 
played over $\ov{G} \times \mem$ by player~2 and player~3 after fixing the strategy 
$\ov{\straa}$ for player~1.
Let 
\[
\ov{U}_1^{\ov{\straa}}=\set{(\ov{s},m) \in \ov{S} \times \mem \mid \text{player~3 has a sure winning 
strategy for the objective $\Parity(\ov{p})$ from $(\ov{s},m)$ in } \ov{G}_{\ov{\straa}}
};
\]
and let $\ov{U}_2^{\ov{\straa}} = (\ov{S}\times \mem) \setminus \ov{U}_1^{\ov{\straa}}$
be the set of sure winning states for player~2 in $\ov{G}_{\ov{\straa}}$.
Consider the strategy $\straa=\tras(\ov{\straa})$, and the sets
$U_1^{\straa}=\set{(s,m) \in S\times \mem \mid (\ov{s},m) \in \ov{U}_1^{\ov{\straa}} }$;
and $U_2^\straa= (S\times \mem) \setminus U_1^\straa$.
The following assertions hold.

\begin{enumerate}
\item For all $(s,m)\in U_1^\straa$, for all strategies $\strab$ of player~2 we 
have  $\Prb_{(s,m)}^{\straa,\strab}(\Parity(p))=1$.
\item For all $(s,m) \in U_2^\straa$, there exists a strategy $\strab$ of player~2
such that $\Prb_{(s,m)}^{\straa,\strab}(\Parity(p))<1$. 
\end{enumerate}
\end{lemma}

We first present the proof for part~1 and then for part~2.

\begin{proof}[(of Lemma~\ref{lemm:par-reduction1}: part~1).]
Consider a finite-memory strategy $\ov{\straa}$ for player~1 with memory 
$\mem$ in the game $\ov{G}$.
Once the strategy $\ov{\straa}$ is fixed we obtain the two-player finite-state
perfect-observation game $\ov{G}_{\ov{\straa}}$ (between player~3 and the adversary player~2).
Recall the sure winning sets
\[
\ov{U}_1^{\ov{\straa}}=\set{(\ov{s},m) \in \ov{S} \times \mem \mid \text{player~3 has a sure winning 
strategy for the objective $\Parity(\ov{p})$ from $(\ov{s},m)$ in } \ov{G}_{\ov{\straa}}
}
\]
for player~3, and 
$\ov{U}_2^{\ov{\straa}} = (\ov{S}\times \mem) \setminus \ov{U}_1^{\ov{\straa}}$
for player~2, respectively, in $\ov{G}_{\ov{\straa}}$.
Let $\straa=\tras(\ov{\straa})$ be the corresponding strategy in $G$.
We denote by $\ov{\straa}^*=\MemLess(\ov{\straa})$ and $\straa^*$ the 
corresponding memoryless strategies of $\ov{\straa}$ in 
$\ov{G} \times \mem$ and $\straa$ in $G \times \mem$,
respectively.
We show that all states in $U_1^{\straa}$ are almost-sure winning, i.e., 
given $\straa$, for all $(s,m) \in U_1^{\straa}$, for all strategies $\strab$ for 
player~2 in $G$ we have $\Prb_{(s,m)}^{\straa,\strab}(\Parity(p))=1$
(recall $U_1^{\straa}=\set{(s,m) \in S\times \mem \mid (\ov{s},m) \in \ov{U}_1^{\ov{\straa}} }$). 
We also consider explicitly the MDP 
$(G \times \mem \obciach U_1^{\straa})_{\straa^*}$ to analyze strategies 
of player~2 on the synchronous product, i.e., we consider the player-2 MDP 
obtained after fixing the memoryless strategy $\straa^*$ in $G\times \mem$,
and then restrict the MDP to the set $U_1^\straa$.

\smallskip\noindent{\em Two key components.}
The proof has two key components.
First, we argue that all end components in the MDP restricted to $U_1^\straa$ 
are winning for player~1 (have min priority even).
Second we argue that given the starting state $(s,m)$ is in $U_1^\straa$, 
almost-surely the set of states visited infinitely often is 
an end component in $U_1^\straa$ 
against all strategies of player~2.
These two key components establish the desired result.

\medskip\noindent{\em Winning end components.} 
Our first goal is to show that every end component $C$ in the player-2 MDP 
$(G\times \mem \obciach U_1^{\straa})_{\straa^*}$ 
is winning for player~1 for the parity objective, i.e.,  
the minimum priority of $C$ is even.
We argue that if there is an end component $C$ in $(G \times \mem \obciach U_1^{\straa})_{\straa^*}$ 
that is winning for player~2 for the parity objective (i.e., minimum priority of $C$ is odd),
then against any memoryless player-3 strategy $\ov{\strac}$ in $\ov{G}_{\ov{\straa}}$, 
player~2 can construct a cycle in the game 
$(\overline{G}\times \mem \obciach \overline{U}_1^{\ov{\straa}})_{\ov{\straa}^*}$ 
that is winning for player~2 (i.e., minimum priority of the cycle is odd) 
(note that given the strategy $\ov{\straa}$ is fixed, we have finite-state
perfect-observation parity games, and hence in the enlarged game we can restrict ourselves
to memoryless strategies for player~3).
This gives a contradiction because player~3 has a sure winning strategy 
from the set $\overline{U}_1^{\ov{\straa}}$ in the 2-player parity game $\overline{G}_{\ov{\straa}}$.
Towards contradiction, let $C$ be an end component in 
$(G\times \mem \obciach U_1^\straa)_{\straa^*}$ that is winning for 
player~2, and let its minimum odd priority be $2r-1$, for some $r \in \Nats$.
Then there is a memoryless strategy $\strab'$ for player~2 in the MDP 
$(G\times \mem \obciach U_1^\straa)_{\straa^*}$ such that $C$ is a bottom 
scc (or a terminal scc) in the Markov chain graph of $(G\times \mem \obciach U_1^\straa)_{\straa^*,\strab'}$. 
Let $\ov{\strac}$ be a memoryless for player~3 in $(\ov{G} \times \mem \obciach \ov{U}_1^{\ov{\straa}})_{\ov{\straa}^*}$. 
Given $\ov{\strac}$ for player~3 and strategy $\strab'$ for player~2 in 
$G \times \mem$, we construct a strategy $\ov{\strab}$ for player~2 in the 
game $(\ov{G} \times \mem \obciach \ov{U}_1^{\ov{\straa}})_{\ov{\straa}^*}$ 
as follows.
For a player-2 state in $C$, the strategy $\ov{\strab}$ follows the strategy 
$\strab'$,
i.e., for a state $(s,{ m}) \in C$ with $s \in \VB$ we have 
$\ov{\strab}((\ov{s},{ m})) =(\ov{t},{ m'})$ where 
$(t,{ m'})=\strab'((s,{ m}))$.
For a probabilistic state in $C$ we define the strategy as follows (i.e., we now 
consider a state $(s,{ m}) \in C$ with $s \in \VR$):
\begin{itemize}
\item if for some successor state $((\wt{s},2\ell),{ m'})$ of $(\ov{s},{ m})$,
the player-3 strategy $\ov{\strac}$ chooses a successor $((\wh{s},2\ell-1),{ m''}) \in C$ 
at the state  $((\wt{s},2\ell),{ m'})$, for $\ell < r$, 
then the strategy $\ov{\strab}$ chooses at state $(\ov{s},{ m})$ the 
successor $((\wt{s}, 2\ell),{ m'})$; and 

\item otherwise the strategy $\ov{\strab}$ chooses at state 
$(\ov{s},{ m})$ the successor $((\wt{s}, 2r),{ m'})$, and
at $((\wh{s}, 2r),{ m''})$ it chooses a successor shortening the
distance (i.e., chooses a successor with smaller breadth-first-search
distance) to a fixed state $(\ov{s}^*,m^*)$ of priority $2r-1$ of $C$ 
(such a state $(s^*,m^*)$ exists in $C$ since $C$ is strongly connected and has 
minimum priority $2r-1$); 
and for the fixed state of priority $2r-1$ the strategy chooses a successor 
$(\ov{s},{m}')$ such that $(s,{m'})\in C$.
\end{itemize}
Consider an arbitrary cycle in the subgraph
$(\overline{G} \times \mem \obciach \ov{C})_{\ov{\straa} , \ov{\strab},\ov{\strac}}$
where $\overline{C}$ is the set of states in the gadgets of states in $C$.   
There are two cases.
\begin{itemize}
\item    
If there is at least one state $((\wh{s}, 2\ell-1),{ m})$,
with $\ell \leq r$ on the cycle, then the minimum priority on the 
cycle is odd, as even priorities smaller than $2r$ are not visited
by the construction as $C$ does not contain states of even 
priorities smaller than $2r$.
\item    
Otherwise, in all states choices shortening the distance to the
state with priority $2r-1$ are taken and hence the cycle must contain a 
priority $2r-1$ state and all other priorities on the cycle are 
$\geq 2r-1$, so $2r-1$ is the minimum priority on the cycle. 
\end{itemize}
Hence a winning end component for player~2 in the MDP 
contradicts that player~3 has a sure winning strategy in $\ov{G}_{\ov{\straa}}$ 
from $\ov{U}_1^{\ov{\straa}}$.
Thus it follows that all end components are winning for player~1 in 
$(G \times \mem \obciach U_1^{\straa})_{\straa^*}$.

\medskip\noindent{\em Almost-sure reachability to winning end-components.}
Finally, we consider the probability of staying in $U_1^\straa$.
For every probabilistic state $(s,m) \in (\VR \times \mem)\cap U_1^\straa$, 
all of its successors must be in $U_1^\straa$.
Otherwise,  player~2 in the state $(\ov{s},m)$ of the game
$\ov{G}_{\ov{\straa}}$ can choose the successor $(\wt{s}, 0)$ and
then a successor to its winning set $\overline{U}_2^{\ov{\straa}}$. 
This again contradicts the assumption that $(\ov{s},m)$ belong to 
the sure winning states $\ov{U}_1^{\ov{\straa}}$ for player~3 in $\ov{G}_{\ov{\straa}}$.
Similarly, for every state $(s,m)\in (S_2\times \mem) \cap U_1^\straa$ 
we must have all its successors are in $U_1^\straa$. 
For all states $(s,m) \in (S_1 \times \mem) \cap U_1^\straa$, the 
strategy $\straa$ chooses a successor in $U_1^\straa$.
Hence for all strategies $\strab$ of player~2, for all states 
$(s,m) \in U_1^\straa$, the objective $\Safe(U_1^\straa)$ 
(which requires that only states in $U_1^\straa$ are visited)
is ensured almost-surely (in fact surely), and hence with 
probability~1 the set of states visited infinitely often is an end component 
in $U_1^\straa$ (by key property of MDPs).
Since every end component in $(G \times \mem \obciach U_1^\straa)_{{\straa^*}}$ 
has even minimum priority, 
it follows that the strategy $\straa$ is an almost-sure winning strategy 
for the parity objective $\Parity(p)$ for player~1 from all states 
$(s,m) \in U_1^\straa$.
This concludes the proof for first part of the lemma.
\hfill\qed
\end{proof}

We now present the proof for the second part.

\begin{proof}[(of Lemma~\ref{lemm:par-reduction1}:part~2).]
Consider a memoryless sure winning strategy $\ov{\strab}$ for player~2 
in $\ov{G}_{\ov{\straa}}$ from the set $\ov{U}_2^{\ov{\straa}}$.
Let us consider the strategies $\straa=\tras(\ov{\straa})$ and 
$\strab=\tras(\ov{\strab})$, and consider the Markov chain 
$G_{\straa,\strab}$.
Our proof shows the following two properties to establish the 
claim: (1)~in the Markov chain $G_{\straa,\strab}$ all bottom sccs (the 
recurrent classes) in $U_2^\straa$ have odd minimum priority;
and (2)~from all states in $U_2^\straa$ some recurrent class in $U_2^\straa$
is reached with positive probability.  
This establishes the desired result of the lemma.

\medskip\noindent{\em No winning bottom scc for player~1 in $U_2^\straa$.}
Assume towards contradiction that there is a bottom scc $C$ contained in 
$U_2^\straa$ in the Markov chain $G_{\straa,\strab}$ 
such that the minimum priority in $C$ is even. 
From $C$ we construct a winning cycle (minimum priority is even) in 
$\ov{U}_2^{\ov{\straa}}$ for player~3 in the game $\ov{G}_{\ov{\straa}}$ 
given the strategy $\ov{\strab}$.
This contradicts that $\ov{\strab}$ is a sure winning strategy for player~2
from $\ov{U}_2^{\ov{\straa}}$ in $\ov{G}_{\ov{\straa}}$. 
Let the minimum priority of $C$ be $2r$ for some $r \in \Nats$.
The idea is similar to the construction of part~1.  
Given $C$,  and  the strategies $\ov{\straa}$ and $\ov{\strab}$, 
we construct a strategy $\ov{\strac}$ for player~3 in $\ov{G}$ as follows: 
For a probabilistic state $(s,m)$ in $C$:
\begin{itemize}
\item if $\ov{\strab}$ chooses
a state $((\wt{s},2\ell-2),m')$, with $\ell\leq r$, then $\ov{\strac}$ 
chooses the successor $((\wh{s},2\ell-2),m')$;
\item otherwise $\ell>r$ (i.e., $\ov{\strab}$ chooses a state 
$((\wt{s},2\ell-2),m')$ for $\ell > r$), then 
$\ov{\strac}$ chooses the state $((\wh{s},2\ell-1),m')$, and 
then a successor to shorten the distance to a fixed state with priority $2r$ 
(such a state exists in $C$); 
and for the fixed state of priority $2r$, the strategy $\ov{\strac}$ chooses 
a successor in $C$.
\end{itemize}
Similar to the proof of part~1, we argue that we obtain a cycle with minimum even priority in the graph
$(\ov{G} \times \mem \obciach \ov{U}_2^{\ov{\straa}})_{\ov{\straa},\ov{\strab},\ov{\strac}}$.
Consider an arbitrary cycle in the subgraph
$(\overline{G} \times \mem \obciach \ov{C})_{\ov{\straa} , \ov{\strab},\ov{\strac}}$
where $\overline{C}$ is the set of states in the gadgets of states in $C$.   
There are two cases.
\begin{itemize}
\item    
If there is at least one state $((\wh{s}, 2\ell-2),{ m})$,
with $\ell \leq r$ on the cycle, then the minimum priority on the 
cycle is even, as odd priorities strictly smaller than $2r+1$ are not visited
by the construction as $C$ does not contain states of odd priorities 
strictly smaller than $2r+1$.
\item    
Otherwise, in all states choices shortening the distance to the
state with priority $2r$ are taken and hence the cycle must contain a 
priority $2r$ state and all other priorities on the cycle are 
$\geq 2r$, so $2r$ is the minimum priority on the cycle. 
\end{itemize}
Thus we obtain cycles winning for player~3, and this contradicts that 
$\ov{\strab}$ is a sure winning strategy for player~2 from 
$\ov{U}_2^{\ov{\straa}}$.
Thus it follows that all recurrent classes in $U_2^\straa$ in the Markov
chain $G_{\straa,\strab}$ are winning for player~2.

\medskip\noindent{\em Not almost-sure reachability to ${U}_1^{{\straa}}$.} 
We now argue that given $\straa$ and $\strab$ there exists no state in 
$U_2^\straa$ such that $U_1^\straa$ is reached almost-surely. 
This would ensure that from all states in $U_2^\straa$ some recurrent class in 
$U_2^\straa$ is reached with positive probability and establish the desired
claim since we have already shown that all recurrent classes in $U_2^\straa$ 
are winning for player~2. 
Given $\straa$ and $\strab$, let $X \subseteq U_2^\straa$ be the set of states 
such that the set $U_1^\straa$ is reached almost-surely from $X$, and assume 
towards contradiction that $X$ is non-empty.
This implies that from every state in $X$, in the Markov chain 
$G_{\straa,\strab}$, there is a path to the set $U_1^\straa$,
and from all states in $X$ the successors are in $X$. 
We construct a strategy $\ov{\strac}$ in the three-player game $\ov{G}_{\ov{\straa}}$
against strategy $\ov{\strab}$ exactly as the strategy constructed for winning bottom scc, 
with the following difference: instead of shortening distance the a fixed 
state of priority $2r$ (as for winning bottom scc's), in this case
the strategy $\ov{\strac}$ shortens distance to $\ov{U}_1^{\ov{\straa}}$.
Formally, given $X$,  the strategies $\ov{\straa}$ and $\ov{\strab}$, 
we construct a strategy $\ov{\strac}$ for player~3 in $\ov{G}$ as follows: 
For a probabilistic state $(s,m)$ in $X$:
\begin{itemize}
\item if $\ov{\strab}$ chooses
a state $((\wt{s},2\ell),m')$, with $\ell\geq 1$, 
then $\ov{\strac}$ chooses the state $((\wh{s},2\ell-1),m')$, 
and then a successor to shorten the distance to the set 
$\ov{U}_1^{\ov{\straa}}$ (such a successor exists since from all states 
in $X$ the set  $\ov{U}_1^{\ov{\straa}}$ is reachable).
\end{itemize}
Against the strategy of player~3 in $\ov{G}_{\ov{\straa}}$ either 
(i)~$\ov{U}_1^{\ov{\straa}}$ is reached in finitely many steps, or 
(ii)~else player~2 infinitely often chooses successor states of the 
form $(\wt{s},0)$ with priority~0 (the minimum even priority), i.e., 
there is a cycle with a state $(\wt{s},0)$ which has priority~0.
If priority~0 is visited infinitely often, then the parity objective is satisfied.
This ensures that in $\ov{G}_{\ov{\straa}}$ player~3 can ensure either to 
reach $\ov{U}_1^{\ov{\straa}}$ in finitely many steps from some state in $\ov{U}_2^{\ov{\straa}}$ 
against $\ov{\strab}$, or the parity objective is satisfied without reaching 
$\ov{U}_1^{\ov{\straa}}$.
In either case this implies that against $\ov{\strab}$ player~3 can ensure to 
satisfy the parity objective (by reaching $\ov{U}_1^{\ov{\straa}}$ in finitely many 
steps and then playing a sure winning strategy from $\ov{U}_1^{\ov{\straa}}$, 
or satisfying the parity objective without reaching $\ov{U}_1^{\ov{\straa}}$ 
by visiting priority~0 infinitely often) 
from some  state in $\ov{U}_2^{\ov{\straa}}$, 
contradicting that $\ov{\strab}$ is a sure winning strategy for player~2 from $\ov{U}_2^{\ov{\straa}}$. 
Thus we have a contradiction, and obtain the desired result.
\hfill\qed
\end{proof}           

\medskip\noindent 
Lemma~\ref{lemm:par-reduction1} establishes the desired correctness result as 
follows:
(1)~If $\ov{\straa}$ is a finite-memory strategy such that in $\ov{G}_{\ov{\straa}}$
player~3 has a sure winning strategy, then by part~1 of Lemma~\ref{lemm:par-reduction1} we obtain 
that $\straa=\tras(\ov{\straa})$ is almost-sure winning.
(2)~Conversely, if $\straa$ is a finite-memory almost-sure winning strategy,
then consider a strategy $\ov{\straa}$ such that $\straa=\tras(\ov{\straa})$ 
(i.e., $\ov{\straa}=\tras^{-1}(\straa)$). 
By part~2 of Lemma~\ref{lemm:par-reduction1}, given the finite-memory strategy $\ov{\straa}$,
player~3 must have a sure winning strategy in $\ov{G}_{\ov{\straa}}$, otherwise
we have a contradiction that $\straa$ is almost-sure winning.
Thus we have the following theorem.

\begin{theorem}[Polynomial reduction]\label{thrm:reduction_complexity}
Given a partial-observation stochastic game graph $G$ with a parity objective 
$\Parity(p)$ for player~1, we construct a three-player game 
$\ov{G}=\tras(G)$ with a parity objective $\Parity(\ov{p})$, 
where player~1 has partial-observation and the other two players have 
perfect-observation, in time $O((n+m)\cdot d)$, where $n$ 
is the number of states of the game, $m$ is the number of transitions, 
and $d$ the number of priorities of the priority function $p$, such that 
the following assertion  holds:
there is a finite-memory almost-sure winning strategy $\straa$ for player~1 
in $G$ iff there exists a finite-memory strategy $\ov{\straa}$ for player~1 
in $\ov{G}$ such that in the game $\ov{G}_{\ov{\straa}}$ obtained given
$\ov{\straa}$, player~3 has a sure winning strategy for $\Parity(\ov{p})$.
The game graph $\tras(G)$ has $O(n\cdot d)$ states, $O(m\cdot d)$ transitions,
and $\ov{p}$ has at most $d+1$ priorities. 
\end{theorem}

\begin{remark}[Positive winning]
We have presented the details of the polynomial reduction for almost-sure winning,
and now we discuss how a very similar reduction works for positive winning.
We explain the key steps, and omit the proof as it is very similar to our proof
for almost-sure winning.
For clarity in presentation we use a priority $-1$ in the reduction, which is the least odd 
priority, and visiting the priority $-1$ infinitely often ensures loosing for player~1.
Note that all priorities can be increased by~2 to ensure that priorities are nonnegative,
but we use the priority~$-1$ as it keeps the changes in the reduction for positive winning 
minimal as compared to almost-sure winning.

\smallskip\noindent{\em Key steps.} 
First we observe that in the reduction gadgets for almost-sure winning, player~2 would never 
choose the leftmost edge to state $(\wt{s},0)$ from $\ov{s}$ in the cycles formed, but only
use them for reachability to cycles. 
Intuitively, the leftmost edge corresponds to edges which must be chosen only finitely often
and ensures positive reachability to the desired end components in the stochastic game.
For positive winning these edges need to be in control of player~3, but must be allowed to 
be taken only finitely often.
Thus for positive winning, the gadget is modified as follows: 
(i)~we omit the leftmost edge from the state $\ov{s}$;
(ii)~we add an additional player-3 state $\wh{s}$ in the beginning, which has an edge to $\ov{s}$ 
and an edge to $(\wh{s},0)$; 
and (iii)~the state $(\wh{s},0)$ is assigned priority~$-1$.
Figure~\ref{fig:gadget-positive} presents a pictorial illustration of the 
gadget of the reduction for positive winning.
Note that in the reduction for positive winning the finite reachability through the leftmost edge
is in control of player-3, but it has the worst odd priority and must be used only finitely often.
This essentially corresponds to reaching winning end components in finitely many steps in the stochastic
game.
In the game obtained after the reduction, the three-player game is surely winning iff player~1 has a 
finite-memory positive winning strategy in the partial-observation stochastic game.
\begin{figure}[!ptb]
  \begin{center}
    \hrule
    \begin{picture}(122,90)(0,0)
\put(-3,0)
{

\gasset{Nw=9,Nh=9,Nmr=4.5,rdist=1, loopdiam=6}


\node[Nmarks=n](ws)(47.5,85){$\widehat{s}$}
\nodelabel[ExtNL=y, NLangle=5, NLdist=1](ws){$p(s)$}

\node[Nmarks=n, Nmr=0](s)(47.5,65){$\ov{s}$}
\nodelabel[ExtNL=y, NLangle=5, NLdist=1](s){$p(s)$}

\node[Nframe=n](dots)(58,58){$\dots$}

\node[Nmarks=n](t0)(10,45){$(\wt{s},0)$}
\nodelabel[ExtNL=y, NLangle=0, NLdist=1](t0){$p(s)$}
\node[Nmarks=n](t2)(32.5,45){$(\wt{s},2)$}
\nodelabel[ExtNL=y, NLangle=0, NLdist=1](t2){$p(s)$}
\node[Nmarks=n](t4)(62.5,45){$(\wt{s},4)$}
\nodelabel[ExtNL=y, NLangle=0, NLdist=1](t4){$p(s)$}

\node[Nframe=n](dots)(84,45){$\dots$}

\node[Nmarks=n, Nadjust=w, Nmr=8](tp)(105,45){$(\wt{s},p(s))$}
\nodelabel[ExtNL=y, NLangle=0, NLdist=1](tp){$p(s)$}

\node[Nmarks=n, Nmr=0](u0)(10,25){$(\wh{s},0)$}
\nodelabel[ExtNL=y, NLangle=0, NLdist=1](u0){$-1$}
\node[Nmarks=n](u1)(25,25){$(\wh{s},1)$}
\nodelabel[ExtNL=y, NLangle=0, NLdist=1](u1){$1$}
\node[Nmarks=n, Nmr=0](u2)(40,25){$(\wh{s},2)$}
\nodelabel[ExtNL=y, NLangle=0, NLdist=1](u2){$2$}
\node[Nmarks=n](u3)(55,25){$(\wh{s},3)$}
\nodelabel[ExtNL=y, NLangle=0, NLdist=1](u3){$3$}
\node[Nmarks=n, Nmr=0](u4)(70,25){$(\wh{s},4)$}
\nodelabel[ExtNL=y, NLangle=0, NLdist=1](u4){$4$}

\node[Nframe=n](dots)(82,25){$\dots$}

\node[Nmarks=n, Nadjust=w](up)(95,25){$(\wh{s},p(s)\!-\!1)$}
\nodelabel[ExtNL=y, NLangle=325, NLdist=0](up){$p(s)\!-\!1$}

\node[Nmarks=n, Nmr=0, Nadjust=w](uq)(120,25){$(\wh{s},p(s))$}
\nodelabel[ExtNL=y, NLangle=325, NLdist=0](uq){$p(s)$}

\node[Nframe=n, Nw=4, Nh=4, Nmarks=n](dummy)(5,5){$\cdot$}
\drawedge[ELpos=50, ELside=r, curvedepth=0](u0,dummy){}

\node[Nframe=n, Nw=5, Nh=5, Nmarks=n](dummy)(15,5){$\cdot$}
\drawedge[ELpos=50, ELside=r, curvedepth=0](u0,dummy){}

\node[Nmarks=n, Nframe=n](label)(10,3){$E(s)$}

\node[Nframe=n, Nw=4, Nh=4, Nmarks=n](dummy)(20,5){$\cdot$}
\drawedge[ELpos=50, ELside=r, curvedepth=0](u1,dummy){}

\node[Nframe=n, Nw=5, Nh=5, Nmarks=n](dummy)(30,5){$\cdot$}
\drawedge[ELpos=50, ELside=r, curvedepth=0](u1,dummy){}

\node[Nmarks=n, Nframe=n](label)(25,3){$E(s)$}

\node[Nframe=n, Nw=4, Nh=4, Nmarks=n](dummy)(35,5){$\cdot$}
\drawedge[ELpos=50, ELside=r, curvedepth=0](u2,dummy){}

\node[Nframe=n, Nw=5, Nh=5, Nmarks=n](dummy)(45,5){$\cdot$}
\drawedge[ELpos=50, ELside=r, curvedepth=0](u2,dummy){}

\node[Nmarks=n, Nframe=n](label)(40,3){$E(s)$}

\node[Nframe=n, Nw=4, Nh=4, Nmarks=n](dummy)(50,5){$\cdot$}
\drawedge[ELpos=50, ELside=r, curvedepth=0](u3,dummy){}

\node[Nframe=n, Nw=5, Nh=5, Nmarks=n](dummy)(60,5){$\cdot$}
\drawedge[ELpos=50, ELside=r, curvedepth=0](u3,dummy){}

\node[Nmarks=n, Nframe=n](label)(55,3){$E(s)$}

\node[Nframe=n, Nw=4, Nh=4, Nmarks=n](dummy)(65,5){$\cdot$}
\drawedge[ELpos=50, ELside=r, curvedepth=0](u4,dummy){}

\node[Nframe=n, Nw=5, Nh=5, Nmarks=n](dummy)(75,5){$\cdot$}
\drawedge[ELpos=50, ELside=r, curvedepth=0](u4,dummy){}

\node[Nmarks=n, Nframe=n](label)(70,3){$E(s)$}

\node[Nframe=n, Nw=4, Nh=4, Nmarks=n](dummy)(90,5){$\cdot$}
\drawedge[ELpos=50, ELside=r, curvedepth=0](up,dummy){}

\node[Nframe=n, Nw=5, Nh=5, Nmarks=n](dummy)(100,5){$\cdot$}
\drawedge[ELpos=50, ELside=r, curvedepth=0](up,dummy){}

\node[Nmarks=n, Nframe=n](label)(95,3){$E(s)$}

\node[Nframe=n, Nw=4, Nh=4, Nmarks=n](dummy)(115,5){$\cdot$}
\drawedge[ELpos=50, ELside=r, curvedepth=0](uq,dummy){}

\node[Nframe=n, Nw=5, Nh=5, Nmarks=n](dummy)(125,5){$\cdot$}
\drawedge[ELpos=50, ELside=r, curvedepth=0](uq,dummy){}

\node[Nmarks=n, Nframe=n](label)(120,3){$E(s)$}

\drawedge[ELpos=50, ELside=r, curvedepth=0](ws,t0){}
\drawedge[ELpos=50, ELside=r, curvedepth=0](ws,s){}

\drawedge[ELpos=50, ELside=l, curvedepth=0](s,t2){}

\drawedge[ELpos=50, ELside=r, curvedepth=0](s,t4){}
\drawedge[ELpos=50, ELside=l, curvedepth=0](s,tp){}

\drawedge[ELpos=50, ELside=r, curvedepth=0](t0,u0){}

\drawedge[ELpos=50, ELside=r, curvedepth=0](t2,u1){}
\drawedge[ELpos=50, ELside=r, curvedepth=0](t2,u2){}

\drawedge[ELpos=50, ELside=r, curvedepth=0](t4,u3){}
\drawedge[ELpos=50, ELside=r, curvedepth=0](t4,u4){}

\drawedge[ELpos=50, ELside=r, curvedepth=0](tp,up){}
\drawedge[ELpos=50, ELside=r, curvedepth=0](tp,uq){}
}

\end{picture}
    \hrule
    \caption{Reduction gadget for positive winning when $p(s)$ is even. \label{fig:gadget-positive}}
  \end{center}
\end{figure}
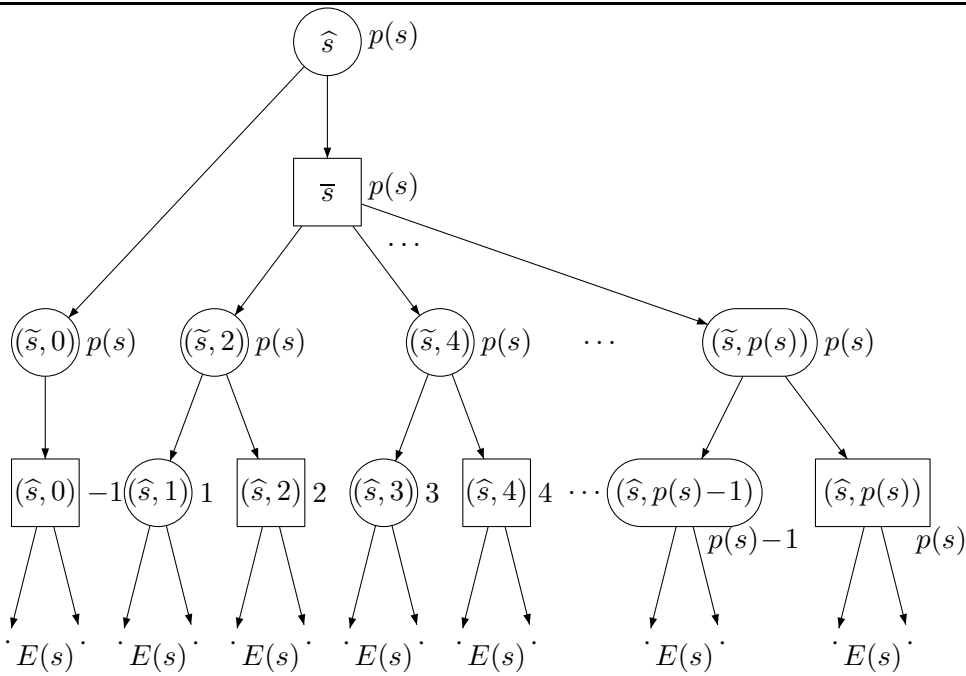
\end{remark}

In this section we established polynomial reductions of the 
qualitative-analysis problems for partial-observation stochastic parity games under
finite-memory strategies to the sure winning problem in three-player games
(player~1 partial, both the other players perfect, and 
player~1 and 3 existential, player~2 adversarial).
The following section shows that the sure winning problem for three-player games
is EXPTIME-complete by reduction to alternating parity tree automata.

\newcommand{\nat}{\mathbb N} 
\newcommand{\A}{\mathcal{A}}
\newcommand{\good}{{\sf good}}
\newcommand{\bad}{{\sf bad}}
\newcommand{\Alphabet}{{\Omega}}
\newcommand{\B}{{\mathcal{B}}}
\newcommand{\N}{\mathcal{N}}

\newcommand{\true}{{\sf true}}
\newcommand{\false}{{\sf false}}

\section{Solving Sure Winning for Three-player Parity Games}
In this section we present the solution for sure winning in three-player
non-stochastic parity games.
We start with the basic definitions.

\subsection{Basic definitions}
We first present a model of partial-observation concurrent three-player games, 
where player~$1$ has partial observation, and player~$2$ and player~$3$ 
have perfect observation. 
Player~$1$ and player~$3$ have the same objective and they play against player~$2$.
We also show that three-player turn-based games model (of 
Section~\ref{sec:partial_stoch}) can be treated as a special case of this model.

\smallskip\noindent{\em Partial-observation three-player concurrent games.}
Given alphabets $A_i$ of actions for player~$i$ ($i=1,2,3$), a partial-observation 
three-player concurrent game (for brevity, \emph{three-player game} in sequel) 
is a tuple $G = \tuple{S, s_0, \delta,\Obs,\obs}$ where:
\begin{itemize}
\item $S$ is a finite set of states;
\item $s_0 \in S$ is the initial state;
\item $\delta: S \times A_1 \times A_2 \times A_3 \to S$ is a deterministic 
transition function that, given a current state $s$, and actions $a_1 \in A_1$, $a_2 \in A_2$, $a_3 \in A_3$
of the players, gives the successor state $s' = \delta(s,a_1,a_2,a_3)$ of $s$; and
\item $\Obs$ is a finite set of observations and $\obs$ is the observation mapping 
(as in Section~\ref{sec:partial_stoch}).
\end{itemize}

\smallskip\noindent{\em Modeling turn-based games.}
A three-player turn-based game is a special case of the model three-player
concurrent games.
Formally, we consider a three-player turn-based game as a tuple 
$\tuple{S_1,S_2,S_3, A_1, \delta, E}$ where $\delta: S_1 \times A_1 \to S_2$ 
is the transition function for player~$1$,
and $E \subseteq (S_2 \cup S_3) \times S$ is a set of edges.
Since player~$2$ and player~$3$ have perfect observation, we consider 
that $A_2 = S$ and $A_3 = S$, that is player~$2$ and player~$3$ choose directly a successor in the game.
The transition function $\overline{\delta}$ for an equivalent concurrent version is as follows
(i)~for $s \in S_1$,  for all $a_2 \in A_2$ and $a_3 \in A_3$, we have 
$\overline{\delta}(s,a_1,a_2,a_3)=\delta(s,a_1)$;
(ii)~for $s \in S_2$,  for all $a_1 \in A_1$ and $a_3 \in A_3$, for $a_2=s'$ 
we have $\overline{\delta}(s,a_1,a_2,a_3)=s'$  if $(s,s')\in E$, else 
$\overline{\delta}(s,a_1,a_2,a_3)=s_{\good}$, where $s_{\good}$ is a special state
in which player~$2$ loses (the objective of player~$1$ and~$3$ is satisfied
if player~$2$ chooses an edge that is not in $E$); 
and  
(iii)~for $s \in S_3$,  for all $a_1 \in A_1$ and $a_2 \in A_2$, for $a_3=s'$ 
we have $\overline{\delta}(s,a_1,a_2,a_3)=s'$  if $(s,s')\in E$, else 
$\overline{\delta}(s,a_1,a_2,a_3)=s_{\bad}$, where $s_{\bad}$ is a special state
in which player~$2$ wins (the objective of player~$1$ and~$3$ is violated
if player~$3$ chooses an edge that is not in $E$).
The set $\Obs$ and the mapping $\obs$ are obvious.

\smallskip\noindent{\em Strategies.}
Define the set $\Straa$ of \emph{strategies} $\straa: \Obs^+ \to A_1$ of player~$1$ that, 
given a sequence of past observations, return an action for player~$1$. Equivalently, we sometimes
view a strategy of player~$1$ as a function $\straa: S^+ \to A_1$ satisfying
$\straa(\rho) = \straa(\rho')$ for all $\rho,\rho' \in S^+$ such that $\obs(\rho) = \obs(\rho')$,
and say that $\straa$ is \emph{observation-based}.
A strategy of player~$2$ (resp, player~$3$) is a function $\strab: S^+ \to A_2$ 
(resp., $\strac: S^+ \to A_3$) without any restriction. 
We denote by $\Strab$ and $\Strac$ the set of strategies of player~$2$ and 
player~$3$, respectively.

\smallskip\noindent{\em Sure winning.}
Given strategies $\straa$, $\strab$, $\strac$ of the three players in $G$, 
the \emph{outcome play} from $s_0$ is the infinite sequence $\rho^{\straa,\strab,\strac}_{s_0} = s_0 s_1 \dots$
such that for all $j \geq 0$, we have 
$s_{j+1} = \delta(s_j,a_j,b_j,c_j)$ where $a_j = \straa(s_0 \dots s_j)$,
$b_j = \strab(s_0 \dots s_j)$, and $c_j = \strac(s_0 \dots s_j)$.
Given a game $G = \tuple{S, s_0, \delta,\Obs,\obs}$ and a parity objective $\varphi \subseteq S^{\omega}$,
the sure winning problem asks to decide if $\exists \straa \in \Straa \cdot
\exists \strac \in \Strac \cdot \forall \strab \in \Strab: 
\rho^{\straa,\strab,\strac}_{s_0} \in \varphi$.
It will follow from our result that if the answer to the sure winning problem
is yes, then there exists a witness finite-memory strategy $\straa$ for 
player~1.

\subsection{Alternating Tree Automata}
In this section we recall the definitions of alternating tree automata,
and present the solution of the sure winning problem for three-player games 
with parity objectives by a reduction to the emptiness problem of alternating tree 
automata with parity acceptance condition. 

\smallskip\noindent{\em Trees.}
%
Given an alphabet $\Alphabet$, 
an $\Alphabet$-labeled tree $(T,V)$ consists of a prefix-closed 
set $T \subseteq \nat^*$ (i.e., if $x\cdot d \in T$ with $x \in \nat^*$
and $d \in \nat$, then $x \in T$), and a mapping $V: T \to \Alphabet$ that assigns
to each node of $T$ a letter in $\Alphabet$.
Given $x \in \nat^*$ and $d \in \nat$ such that $x \cdot d \in T$, 
we call $x \cdot d$ the \emph{successor} in direction $d$ of $x$. 
The node~$\varepsilon$ is the \emph{root} of the tree.
An \emph{infinite path} in $T$ is an infinite sequence $\pi = d_1 d_2 \dots$
of directions $d_i \in \nat$ such that every finite prefix of $\pi$ is a node in~$T$.

\smallskip\noindent{\em Alternating tree automata.}
Given a parameter $k \in \nat \setminus \{0\}$, we consider input trees of rank 
$k$, i.e. trees in which every node has at most $k$ successors. Let $[k] = \{0,\dots,k-1\}$,
and given a finite set $U$, let $\B^+(U)$ be the set of positive Boolean formulas 
over $U$, i.e. formulas built from elements in $U \cup \{\true,\false\}$ using 
the Boolean connectives $\land$ and $\lor$.
%
An \emph{alternating tree automaton} over alphabet $\Alphabet$ is a tuple 
$\A = \tuple{S,s_0,\delta}$ where:
\begin{itemize}
\item $S$ is a finite set of states;
\item $s_0 \in S$ is the initial state;
\item $\delta: S \times \Alphabet \to \B^+(S \times [k])$ is a transition function.
\end{itemize}

Intuitively, the automaton is executed from the initial state $s_0$ and reads
the input tree in a top-down fashion starting from the root~$\varepsilon$. In state $s$,
if $a \in \Alphabet$ is the letter that labels the current node $x$ of the input tree,
the behavior of the automaton is given by the formulas $\psi = \delta(s,a)$. 
The automaton chooses a \emph{satisfying assignment} of $\psi$,
i.e. a set $Q \subseteq S \times [k]$ such that the formula $\psi$ is satisfied when
the elements of $Q$ are replaced by $\true$, and the elements of $(S \times [k]) \setminus Q$
are replaced by $\false$. Then, for each $\tuple{s_1,d_1} \in Q$ a copy of the automaton is spawned 
in state $s_1$, and proceeds to the node $x \cdot d_1$ of the input tree. 
In particular, it requires that $x \cdot d_1$ belongs to the input tree. 
For example, if $\delta(s,a) = (\tuple{s_1,0} \land \tuple{s_2,0}) \lor (\tuple{s_3,0} \land \tuple{s_4,1} \land \tuple{s_5,1})$,
then the automaton should either spawn two copies that process the successor
of~$x$ in direction~$0$~(i.e., the node $x \cdot 0$) and that enter the respective 
states~$s_1$ and~$s_2$, or spawn three copies of which one processes $x \cdot 0$
and enters state~$s_3$, and the other two process $x \cdot 1$ and enter 
the states~$s_4$ and~$s_5$ respectively. 


\smallskip\noindent{\em Runs.}
A run of $\A$ over an $\Alphabet$-labeled input tree $(T,V)$ is a tree $(T_r,r)$
labeled by elements of $T \times S$, where a node of $T_r$ labeled by $(x,s)$
corresponds to a copy of the automaton proceeding the node~$x$ of the 
input tree in state $s$. Formally, a \emph{run} of $\A$ over an input tree $(T,V)$ is a 
$(T \times S)$-labeled tree $(T_r,r)$ such that $r(\varepsilon) = (\varepsilon,s_0)$
and for all $y \in T_r$, if $r(y) = (x,s)$,
then the set $\{\tuple{s',d'} \mid \exists d \in \nat: r(y \cdot d) = (x \cdot d', s')\}$ 
is a satisfying assignment for $\delta(s,V(x))$.
Hence we require that, given a node $y$ in $T_r$ labeled by $(x,s)$, there is a
satisfying assignment $Q \subseteq S \times [k]$ for the formula $\delta(s,a)$ where
$a = V(x)$ is the letter labeling the current node $x$ of the input tree, 
and for all states $\tuple{s',d'} \in Q$ there is a (successor) node $y \cdot d$ in 
$T_r$ labeled by $(x \cdot d', s')$.

Given an accepting condition $\varphi \subseteq S^{\omega}$, we say that a run $(T_r,r)$
is \emph{accepting} if for all infinite paths $d_1 d_2 \dots$ of $T_r$, 
the sequence $s_1 s_2 \dots$ such that $r(d_i) = (\cdot,s_i)$ for all $i\geq 0$
is in $\varphi$. 
The \emph{language} of $\A$ is the set $L_k(\A)$ of all input trees of rank $k$ 
over which there exists an accepting run of $\A$. The emptiness problem for alternating
tree automata is to decide, given $\A$ and parameter $k$, whether $L_k(\A) = \emptyset$.


\subsection{Solution of the Sure Winning Problem for Three-player Games}
We now present the solution of the sure winning problem for three-player games.

\begin{theorem}
Given a three-player game $G = \tuple{S, s_0, \delta,\Obs,\obs}$ and a \{safety, reachability, parity\} 
objective $\varphi$, the problem of deciding whether 
$$\exists \straa \in \Straa \cdot \exists \strac \in \Strac \cdot \forall \strab \in \Strab : \rho^{\straa,\strab,\strac}_{s_0} \in \varphi$$
is EXPTIME-complete.
\end{theorem}

\begin{proof}
The EXPTIME-hardness follows from EXPTIME-hardness of two-player partial-observation
games with reachability objective~\cite{Reif84,CDHR07} and safety objective~\cite{BD08}.

We prove membership in EXPTIME by a reduction to the emptiness problem for alternating
tree automata, which is solvable in EXPTIME for parity objectives~\cite{MSS86,MS87,MS95}.
The reduction is as follows. Given a game $G = \tuple{S, s_0, \delta,\Obs,\obs}$ 
over alphabet of actions $A_i$ ($i=1,2,3$), 
we construct the alternating tree automaton $\A = \tuple{S',s'_0,\delta'}$ over alphabet $\Alphabet$
and parameter $k = \abs{\Obs}$ (we assume that $\Obs = [k]$) where:
\begin{itemize}
\item $S' = S$, and $s'_0 = s_0$;
\item $\Alphabet = A_1$;
\item $\delta'$ is defined by $\delta'(s,a_1) = \bigvee_{a_3 \in A_3} \bigwedge_{a_2 \in A_2} \tuple{\delta(s,a_1,a_2,a_3), \obs(\delta(s,a_1,a_2,a_3))}$
for all $s \in S$ and $a_1 \in \Alphabet$.
\end{itemize}

The acceptance condition $\varphi$ of the automaton is the same as the objective of the game~$G$.
We prove that $\exists \straa \in \Straa \cdot \exists \strac \in \Strac
\cdot \forall \strab \in \Strab: \rho^{\straa,\strab,\strac}_{s_0} \in \varphi$
if and only if $L_k(\A) \neq \emptyset$. 
We use the following notation. Given a node $y = d_1 d_2 \dots d_n$ in a $(T \times S)$-labeled
tree $(T_r,r)$, consider the prefixes $y_0 = \varepsilon$, and $y_i = d_1 d_2 \dots d_i$ (for $i=1,\dots,n$).
Let $\ov{r}_2(y) = s_0 s_1 \dots s_n$ where $r(y_i) = (\cdot, s_i)$ for $0 \leq i \leq n$, denote 
the corresponding state sequence of $y$.

\begin{enumerate}
\item \emph{Sure winning implies non-emptiness.}
First, assume that for some $\straa \in \Straa$ and $\strac \in \Strac$,
we have $\forall \strab \in \Strab: \rho^{\straa,\strab,\strac}_{s_0} \in \varphi$.
From $\straa$, we define an input tree $(T,V)$ where $T = [k]^*$
and $V(\gamma) = \straa(\obs(s_0)\cdot \gamma)$ for all $\gamma \in T$ (we view
$\straa$ as a function $[k]^+ \to \Alphabet$, since $[k] = \Obs$ and $\Alphabet = A_1$).
From $\strac$, we define a $(T \times S)$-labeled tree $(T_r,r)$ such that
$r(\varepsilon) = (\varepsilon,s_0)$ and for all $y \in T_r$, if $r(y) = (x,s)$ and $\ov{r}_2(y) = \rho$,
then for $a_1 = \straa(\obs(s_0) \cdot x) = V(x)$, for $a_3 = \strac(s_0 \cdot \rho)$,
for every $s'$ in the set $Q = \{s' \mid \exists a_2 \in A_2: s' = \delta(s,a_1,a_2,a_3)\}$,
there is a successor $y \cdot d$ of $y$ in $T_r$ labeled by $r(y \cdot d) = (x \cdot \obs(s'), s')$.
Note that $\{\tuple{s',\obs(s')} \mid s' \in Q\}$ is a satisfying assignment for $\delta'(s, a_1)$ 
and $a_1 = V(x)$, hence 
$(T_r,r)$ is a run of $\A$ over $(T,V)$. For every infinite path $\rho$ in $(T_r,r)$, consider
a strategy $\strab \in \Strab$ consistent with $\rho$. Then $\rho = \rho^{\straa,\strab,\strac}_{s_0}$,
hence $\rho \in \varphi$ and the run $(T_r,r)$ is accepting, showing that $L_k(\A) \neq \emptyset$.

\item \emph{Non-emptiness implies sure winning.}
Second, assume that $L_k(\A) \neq \emptyset$. 
Let $(T,V) \in L_k(\A)$ and $(T_r,r)$ be an accepting run of $\A$ over $(T,V)$.
From $(T,V)$, define a strategy $\straa$ of player~$1$ such that 
$\straa(s_0 \cdot \rho) = V(\obs(\rho))$ for all $\rho \in S^*$. 
Note that $\straa$ is indeed observation-based. 
From $(T_r,r)$, we know that for all nodes $y \in T_r$ with $r(y) = (x, s)$
and $\ov{r}_2(y) = \rho$,
the set $Q = \{\tuple{s',d'} \mid \exists d \in \nat: r(y\cdot d) = (x \cdot d', s')\}$ 
is a satisfying assignment of $\delta'(s, V(x))$, hence there exists 
$a_3 \in A_3$ such that for all $a_2 \in A_2$, there is a successor of $y$ 
labeled by $(x \cdot \obs(s'), s')$ with $s' = \delta(s,a_1,a_2,a_3)$ and $a_1 = \straa(s_0 \cdot \rho)$.
Then define $\strac(s_0 \cdot \rho) = a_3$.
Now, for all strategies $\strab \in \Strab$ the outcome $\rho^{\straa,\strab,\strac}_{s_0}$
is a path in $(T_r,r)$, and hence $\rho^{\straa,\strab,\strac}_{s_0} \in \varphi$.
Therefore $\exists \straa \in \Straa \cdot \exists \strac \in \Strac \cdot
\forall \strab \in \Strab: \rho^{\straa,\strab,\strac}_{s_0} \in \varphi$.
\end{enumerate}
The desired result follows.
\qed
\end{proof}

The nonemptiness problem for an alternating tree automaton $\A$ with 
parity condition can be solved by constructing an equivalent 
nondeterministic parity tree automaton $\N$ (such that $L_k(\A) = 
L_k(\N)$), and then checking emptiness of $\N$. The construction proceeds 
as follows~\cite{MS95}. The nondeterministic automaton $\N$ guess a 
labeling of the input tree with a memoryless strategy for the alternating 
automaton $\A$. As $\A$ has $n$ states and $k$ directions, there are $(k^n)$ 
possible strategies. A nondeterministic parity word automaton with $n$ 
states and $d$ priorities can check that the strategy works along every 
branch of the tree. An equivalent deterministic parity word automaton can 
be constructed with $(n^n)$ states and $O(d\cdot n)$ priorities~\cite{CZ13}. 
Thus, $\N$ can guess the strategy labeling and check the strategies with 
$O((k\cdot n)^n)$ states and $O(d\cdot n)$ priorities. The nonemptiness of $\N$ 
can then be checked by considering it as a 
(two-player perfect-information deterministic)
parity game with $O((k\cdot n)^n)$ 
states and $O(d\cdot n)$ priorities~\cite{GH82}. This games can be solved in 
time $O((k\cdot n)^{d\cdot n^2})$~\cite{EJ91}. Moreover, since memoryless strategies 
exist for parity games~\cite{EJ91}, if the nondeterministic parity tree 
automaton is nonempty, then it accepts a regular tree that can be encoded 
by a transducer with $((k\cdot n)^n)$ states. Thus, the nonemptiness problem 
for alternating tree automaton with parity condition can be decided in 
exponential time, and there exists a transducer to witness nonemptiness that 
has exponentially many states.

\begin{theorem}
Given a three-player game $G = \tuple{S, s_0, \delta,\Obs,\obs}$ with $n$ states (and $k \leq n$ observations for player~$1$) 
and parity objective $\varphi$ defined by $d$ priorities, 
the problem of deciding whether 
$$\exists \straa \in \Straa  \cdot \exists \strac \in \Strac \cdot \forall \strab \in \Strab: \rho^{\straa,\strab,\strac}_{s_0} \in \varphi$$
can be solved in time exponential time.
Moreover, memory of exponential size 
is sufficient for player~$1$.
\end{theorem}

\begin{remark}
By our reduction to alternating parity tree automata and the fact that if
an alternating parity tree automaton is non-empty, there is a regular witness
tree for non-emptiness it follows that strategies for player~1 can be restricted
to finite-memory without loss of generality.
This ensures that we can solve the problem of the existence of 
finite-memory almost-sure winning (resp. positive winning) strategies in 
partial-observation stochastic parity games (by Theorem~\ref{thrm:reduction_complexity} 
of Section~\ref{sec:partial_stoch}) also in EXPTIME,
and EXPTIME-completeness of the problem follows since the problem 
is EXPTIME-hard even for reachability objectives for almost-sure winning~\cite{CDHR07}
and safety objectives for positive winning~\cite{CDH10a}.
\end{remark}

\begin{theorem}\label{thrm:complexity}
Given a partial-observation stochastic game and a parity objective $\varphi$ defined by $d$ priorities, 
the problem of deciding whether there exists a finite-memory almost-sure (resp. positive) winning 
strategy for player~1 
is EXPTIME-complete.
Moreover, if there is an almost-sure (resp. positive) winning strategy, then 
there exists one that uses memory of at most exponential size.
\end{theorem}

\begin{remark}
As mentioned in Remark~\ref{remark:strategies} the EXPTIME upper bound for
qualitative analysis of partial-observation stochastic parity games with 
finite-memory randomized strategies follows from Theorem~\ref{thrm:complexity}.
The EXPTIME lower bound and the exponential lower bound on memory requirement 
for finite-memory randomized strategies follows from the results 
of~\cite{CDHR07,CDH10a} for reachability and safety objectives 
(even for POMDPs). 
\end{remark}

\clearpage
\bibliographystyle{plain}
\bibliography{biblio}

\end{document}